\documentclass[11pt]{article}
\usepackage{amsmath,amsthm,amssymb}
\usepackage{authblk}
\usepackage{charter,eulervm}
\usepackage{graphicx}
\usepackage[a4paper, total={6in, 9.5in}]{geometry}

\newcommand{\tw}{{\rm tw}}
\newcommand{\bs}{{\rm bs}}
\theoremstyle{definition}
\newtheorem{theorem}{Theorem}
\newtheorem{lemma}{Lemma}
\newtheorem{corollary}{Corollary}
\newtheorem{definition}{Definition}
\title{A (probably) optimal algorithm for {\sc Bisection} on bounded-treewidth graphs}

\author[1]{Tesshu Hanaka}
\affil[1]{Chuo University}
\author[2]{Yasuaki Kobayashi}
\author[2]{Taiga Sone}
\affil[2]{Kyoto University}

\begin{document}
\maketitle
\begin{abstract}
The maximum/minimum bisection problems are, given an edge-weighted graph, to find a bipartition of the vertex set into two sets whose sizes differ by at most one, such that the total weight of edges between the two sets is maximized/minimized.
Although these two problems are known to be NP-hard, there is an efficient algorithm for bounded-treewidth graphs.
In particular, Jansen et al. (SIAM J. Comput. 2005) gave an $O(2^tn^3)$-time algorithm when given a tree decomposition of width $t$ of the input graph, where $n$ is the number of vertices of the input graph.
Eiben et al. (ESA 2019) improved the dependency of $n$ in the running time by giving an $O(8^tt^5n^2\log n)$-time algorithm.
Moreover, they showed that there is no $O(n^{2-\varepsilon})$-time algorithm for trees under some reasonable complexity assumption.

In this paper, we show an $O(2^t(tn)^2)$-time algorithm for both problems, which is asymptotically tight to their conditional lower bound. 
We also show that the exponential dependency of the treewidth is asymptotically optimal under the Strong Exponential Time Hypothesis.
Finally, we discuss the (in)tractability of both problems with respect to special graph classes.
\end{abstract}


\section{Introduction}

Let $G = (V, E)$ be a graph and let $w: E \to \mathbb R$ be an edge-weight function.
For disjoint subsets $X, Y$ of $V$, we denote by $w(X, Y)$ the total weight of edges between $X$ and $Y$.
A {\em bisection} of $G$ is a bipartition of $V$ into two sets $A$ and $B$ such that $-1 \le |A| - |B| \le 1$.
The {\em size} of a bisection $(A, B)$ is defined as the number of edges between $A$ and $B$.
We also consider bisections of edge-weighted graphs.
In this case, the size of a bisection $(A, B)$ is defined as $w(A, B)$.

In this paper, we consider the following two problems: \textsc{Min Bisection} and \textsc{Max Bisection}.

\begin{definition}
    Given an edge-weighted graph $G = (V, E)$ with $w: E \to \mathbb R_+$, the problem \textsc{Min Bisection} (resp. \textsc{Max Bisection}) asks for a minimum size (resp. maximum size) bisection $(A, B)$ of $G$.
\end{definition}

These problems are well-known variants of {\sc MinCut} and {\sc MaxCut}, which feasible solutions are not required to be ``balanced''.
If every edge has non-negative weight, {\sc MinCut}, the problem of minimizing $w(A, B)$
over all bipartitions $(A, B)$ of $V$, can be solved in polynomial time.
For {\sc MaxCut}, the maximization version of {\sc MinCut}, the problem is NP-hard in general~\cite{Karp::1972} and trivially solvable in polynomial time for bipartite graphs.
Orlova and Dorfman~\cite{Orlova::1972} and Hadlock~\cite{Hadlock::1975} proved that {\sc MaxCut} can be solved in polynomial time for planar graphs with non-negative edge weights, and Shih et al.~\cite{Shih::1990} finally gave a polynomial-time algorithm for planar graphs with arbitrary edge weights.
In contrast to these complexity status of {\sc MinCut} and {\sc MaxCut}, the bisection problems are particularly hard.
{\sc Max Bisection} is known to be NP-hard even on planar graphs~\cite{Jansen::2005} and unit disk graphs~\cite{Diaz::2007}.
For {\sc Min Bisection}, it is NP-hard \cite{Garey::1976} even on $d$-regular graphs for fixed $d \ge 3$~\cite{Bui::1987} and unit disk graphs~\cite{Diaz::2017}. 
It is worth noting that the complexity of {\sc Min Bisection} on planar graphs is still open.

On bounded-treewidth graphs, {\sc Min Bisection} and {\sc Max Bisection} are solvable in polynomial time.
More precisely, given a graph $G$ of $n$ vertices and a tree decomposition of $G$ of width $t$, Jansen et al.~\cite{Jansen::2005} proved that {\sc Max Bisection} can be solved in time $O(2^{t}n^3)$.
This algorithm also works on graphs with arbitrary edge-weights, which means that {\sc Min Bisection} can be solved within the same running time.
Very recently, Eiben et al. \cite{Eiben::2019} improved the polynomial factor of $n$ by giving an $O(8^tt^5n^2\log n)$-time algorithm.
They also discussed a conditional lower bound on the running time: For any $\varepsilon > 0$, {\sc Min Bisection} cannot be solved in time $O(n^{2 - \varepsilon})$ on $n$-vertex trees unless {\sc $(\min, +)$-Convolution} (defined in Section~\ref{ssec:obtimal}) has an $O(n^{2 - \delta})$-time algorithm for some $\delta > 0$.
Since trees are precisely connected graphs of treewidth at most one, this lower bound also holds for bounded-treewidth graphs.
However, there is still a gap between the upper and (conditional) lower bound on the running time for bounded-treewidth graphs. 

In this paper, we fill this gap by showing an ``optimal'' algorithm for {\sc Min Bisection} and {\sc Max Bisection} on bounded-treewidth graphs.
The running time of our algorithm is $O(2^t(tn)^2)$, provided that a width-$t$ tree decomposition of the input graph is given as input. 
The polynomial factor in $n$ matches the conditional lower bound due to Eiben el al.~\cite{Eiben::2019}.
We also observe that {\sc Max Bisection} cannot be solved in time $(2 - \varepsilon)^{t}n^{O(1)}$ for any $\varepsilon > 0$ unless the Strong Exponential Time Hypothesis (SETH)~\cite{Impagliazzo::2001} fails.
These facts imply that the exponential dependency with respect to $t$ and the polynomial dependency with respect to $n$ in our running time are asymptotically optimal under these well-studied complexity-theoretic assumptions.

To overcome this ``quadratic hurdle'', we consider the vertex cover number of input graphs.
Since the treewidth of a graph is upper bounded by its vertex cover number,
this immediately implies an $O(2^k(kn)^2)$-time algorithm for {\sc Min Bisection} and {\sc Max Bisection}, where $k$ is the vertex cover number of the input graph.
We improve this running time to $O(2^kkn)$, that is, a linear time algorithm for graphs of bounded vertex cover number.

We also investigate the complexity of {\sc Min Bisection} and {\sc Max Bisection} from the viewpoint of special graph classes.
From the known hardness result of {\sc MaxCut}, we immediately have several complexity results for {\sc Max Bisection} and {\sc Min Bisection}.
The most notable case is that both problems are NP-hard even on unweighted bipartite graphs, on which {\sc MaxCut} can be trivially solved in polynomial time.
Apart from these complexity results, we show that {\sc Max Bisection} can be solved in linear time on line graphs.

\paragraph{Difference from the conference version}
Compared to the conference version~\cite{Hanaka:Optimal:2020}, the current paper additionally contains the conditional lower bound for paths based on {\sc $(\min, +)$-Convolution}, which slightly strengthens the conditional lower bound given by \cite{Eiben::2019} and a linear-time algorithm for graphs of bounded vertex cover number.
These are presented in Sections~\ref{ssec:obtimal}~and~\ref{sec:vc}.

\section{Preliminaries}\label{sec:prel}
Let $G = (V, E)$ be a graph, which is simple and undirected.
Throughout the paper, we use $n$ to denote the number of vertices of an input graph.
We also write $V(G)$ to denote the set of vertices of $G$.
For a vertex $v \in V$, we denote by $N(v)$ the set of neighbors of $v$ in $G$.
For two disjoint subsets $X, Y \subseteq V$, we denote by $E(X, Y)$ the set of edges having one end in $X$ and the other end in $Y$.
Let $w: E \to \mathbb R$ be an edge-weight function.
We write $w(X, Y)$ to denote the total weight of edges in $E(X, Y)$ (i.e., $w(X, Y) = \sum_{e \in E(X, Y)}w(e)$).
A bipartition $(A, B)$ of $V$ is called a {\em cut} of $G$.
The {\em size} of a cut is the number of edges between $A$ and $B$, that is, $|E(A, B)|$.
For edge-weighted graphs, the size is measured by the total weight of edges between $A$ and $B$.
A cut is called a {\em bisection} if $-1 \le |A| - |B| \le 1$.

In the next section, we work on dynamic programming based on {\em tree decompositions}.
A {\em tree decomposition} of $G$ is a pair of a rooted tree $T$ with vertex set $I$ and a collection $\{X_i : i \in I\}$ of subsets of $V$  such that
\begin{itemize}
    \item $\bigcup_{i \in I}X_i = V$;
    \item for each $\{u, v\} \in E$, there is an $i \in I$ with $\{u, v\} \subseteq X_i$;
    \item for each $v \in V$, the subgraph of $T$ induced by $\{i \in I : v \in X_i\}$ is connected.  
\end{itemize}
We refer to vertices of $T$ as {\em nodes} to distinguish them from vertices of $G$.
We say that $T$ is a {\em path decomposition} of $G$ if $T$ forms a path.
The {\em width} of $T$ is defined as $\max_{i \in I} |X_i| -1$.
The {\em treewidth} of $G$ is the minimum integer $k$ such that $G$ has a tree decomposition of width $k$ and the {\em pathwidth} of $G$ is the minimum integer $k$ such that $G$ has a path decomposition of width $k$. 

To facilitate dynamic programming on tree decompositions, several types of ``special'' tree decompositions are known.
Jansen et al.~\cite{Jansen::2005} used the well-known {\em nice tree decomposition} for solving {\sc Max Bisection}.
Eiben et al.~\cite{Eiben::2019} improved the dependency on $n$ by means of ``shallow'' tree decompositions due to Bodlaender and Hagerup~\cite{Bodlaender::1998}.
In this paper, we rather use nice tree decompositions as well as Jansen et al.~\cite{Jansen::2005}, and the algorithm itself is, in fact, identical with theirs.

We say that a tree decomposition $T$ is {\em nice} if for every non-leaf node $i$ of $T$,
either
\begin{itemize}
    \item {\bf Introduce node} $i$ has an exactly one child $j \in I$ such that $X_i = X_j \cup \{v\}$ for some $v \in V \setminus X_j$,
    \item {\bf Forget node} $i$ has an exactly one child $j \in I$ such that $X_j = X_i \cup \{v\}$ for some $v \in V \setminus X_i$, or
    \item {\bf Join node} $i$ has exactly two children $j, k \in I$ such that $X_i = X_j = X_k$.
\end{itemize}

\begin{lemma}[Lemma 13.1.3 in \cite{Kloks::1994}]
    Given a tree decomposition of $G$ of width $t$, there is an algorithm that converts it into a nice tree decomposition of width at most $t$ in time $O(t^2n)$.
    Moreover, the constructed nice tree decomposition has at most $4n$ nodes.
\end{lemma}

\section{Bounded-treewidth graphs}\label{sec:btw}
Let $G = (V, E)$ be an edge-weighted graph with weight function $w: E \rightarrow \mathbb R$.
Note that we do not restrict the weight function to take non-negative values.
In this context, {\sc Min Bisection} is essentially equivalent to {\sc Max Bisection}.
Therefore, in this section, we will only consider the maximization counterpart.

\subsection{An $O(2^tn^3)$-time algorithm}
We quickly review the algorithm of Jansen et al.~\cite{Jansen::2005} for computing {\sc Max Bisection} on bounded-treewidth graphs.
Let $T$ be a nice tree decomposition of width at most $t$.
For each node $i \in I$, we use $V_i$ to denote the set of vertices of $G$ that is contained in $X_i$ or $X_j$ for some descendant $j \in I$ of $i$.

Let $i \in I$ be a node of $T$.
For each $S \subseteq X_i$ and $0 \le d \le |V_i|$,
we compute the value $\bs(i, S, d)$ which is the maximum size of a bisection $(A_i, B_i)$ of $G[V_i]$ such that $A_i \cap X_i = S$ and $|A_i| = d$.

\paragraph{\bf{Leaf node}}
Let $i \in I$ be a leaf of $T$. For each $S \subseteq X_i$, $\bs(i, S, d) = w(S, X_i \setminus S)$ if $d = |S|$.
Otherwise we set $\bs(i, S, d) = -\infty$.

\paragraph{\bf{Introduce node}}
Let $i \in I$ be an introduce node of $T$ and let $v \in X_i \setminus X_j$ be the vertex introduced at $i$, where $j \in I$ is the unique child of $i$.
Since the neighborhood of $v$ in $G[V_i]$ is entirely contained in $X_i$,
we can compute $\bs(i, S, d)$ as
\begin{equation*}
    \bs(i, S, d) = \begin{cases}
        \bs(j, S \setminus \{v\}, d - 1) + w(\{v\}, X_i \setminus S) & \text{if } v \in S\\
        \bs(j, S, d) + w(\{v\}, S) & \text{otherwise},
    \end{cases}
\end{equation*}
for each $S \subseteq X_i$ and $0 \le d \le |V_i|$.

\paragraph{\bf{Forget node}}
Let $i \in I$ be a forget node of $T$ and let $v \in X_j \setminus X_i$ be the vertex forgotten at $i$, where $j \in I$ is the unique child of $i$.
As $G[V_i] = G[V_j]$, we can compute $\bs(i, S, d)$ as
\begin{equation*}
    \bs(i, S, d) = \max(\bs(j, S, d), \bs(j, S \cup \{v\}, d))
\end{equation*}
for each $S \subseteq X_i$ and $0 \le d \le |V_i|$.

\paragraph{\bf{Join node}}
Let $i \in I$ be a join node of $T$ with children $j, k \in I$.
By the definition of nice tree decompositions, we have $X_i = X_j = X_k$.
For $S \subseteq X_i$ and $0 \le d \le |V_i|$,
\begin{equation}\label{rec:join}
    \bs(i, S, d) = \max_{|S| \le d' \le d}(\bs(j, S, d') + \bs(k, S, d-d'+|S|) - w(S, X_i \setminus S)).
\end{equation}
Note that the edges between $S$ and $X_i \setminus S$ contribute to both $\bs(j, S, d')$ and $\bs(k, S, d-d'+|S|)$.
Thus, we subtract $w(S, X_i \setminus S)$ in the recurrence (\ref{rec:join}).

\paragraph{\bf{Running time}} 
For each leaf, introduce, or forget node $i$, we can compute $\bs(i, S, d)$ in total time $O(2^tt|V_i|)$ for all $S\subseteq X_i$ and $0 \le d \le |V_i|$.
For join node $i$, the recurrence (\ref{rec:join}) can be evaluated in time $O(|V_i|)$ for each $S \subseteq X_i$ and $0 \le d \le |V_i|$, provided that $\bs(j, S, \ast)$, $\bs(k, S, \ast)$, and $w(S, X_i \setminus S)$ are stored in the table.
Therefore, the total running time for a join node $i$ is $O(2^t|V_i|^2)$.
Since $|V_i| = O(n)$ and $T$ has $O(n)$ nodes, the total running time of the entire algorithm is $O(2^tn^3)$.

\begin{theorem}[\cite{Jansen::2005}]\label{thm:cubic}
    Given a tree decomposition of $G$ of width $t$, {\sc Max Bisection} can be solved in $O(2^tn^3)$ time.
\end{theorem}

\subsection{A refined analysis for join nodes}
As we have seen in the previous subsection, the bottleneck of the algorithm of Theorem~\ref{thm:cubic} appears in computing join nodes.
For a refined running time analysis, we reconsider the recurrence (\ref{rec:join}) for join nodes.
This can be rewritten as
\begin{eqnarray*}
    \bs(i, S, d) &=& \max_{|S| \le d' \le d}(\bs(j, S, d') + \bs(k, S, d-d'+|S|) - w(S, X_i \setminus S))\\
    &=& \max_{\substack{d',d''\\ d' + d'' = d + |S|}} (\bs(j, S, d') + \bs(k, S, d'') - w(S, X_i \setminus S)).
\end{eqnarray*}
Since $d'$ and $d''$ respectively run over $0 \le d' \le |V_j|$ and $0 \le d'' \le |V_k|$, we can compute $\bs(i, S, d)$ in total time $O(2^t|V_j|\cdot|V_k|)$ for all $S$ and $d$.

For each node $i \in I$, we let $n_i = \sum_{j \preceq_T i}|X_j|$, where the summation is taken over all descendants $j$ of $i$ and $i$ itself.
Clearly, $n_i \ge |V_i|$ and hence the total running time of join nodes is upper bounded by
\begin{equation*}
    \sum_{i : \text{ join node}}O(2^tn_jn_k) = O\left(2^t \cdot\sum_{i: \text{ join node}}n_jn_k\right).
\end{equation*}
We abuse the notations $n_j$ and $n_k$ for different join nodes $i$, and the children nodes $j$ and $k$ are defined accordingly.
We claim that $\sum_{i: \text{ join node}}n_jn_k$ is $O((tn)^2)$.
To see this, let us consider the term $n_jn_k$ for a join node $i$.
For each node $q$ of $T$, we label all the vertices contained in $X_q$ by distinct labels $v_1^q, v_2^q, \ldots v_{|X_q|}^q$. 
Note that some vertex can receive two or more labels in this process since a vertex can be contained in more than one node in the tree decomposition. 
From now on, we regard such a vertex as distinct labeled vertices and hence $n_i$ corresponds to the number of labeled vertices that appear in the node $i$ or some descendant node of $i$.
Now, the term $n_jn_k$ can be seen as the number of pairs of labeled vertices $(\ell, r)$ such that $\ell$ is a labeled vertex contained in the subtree rooted at the left child $j$ and $r$ is a labeled vertex contained in the subtree rooted at the right child $k$.
A crucial observation is that any pair of labeled vertices $(\ell, r)$ is counted at most once at the lowest common ancestor of nodes containing $\ell$ and $r$.
This implies that $\sum_{i: \text{ join node}}n_jn_k$ is at most the number of distinct pairs of labeled vertices.
Since each node of $T$ contains at most $t + 1$ vertices and $T$ contains $O(n)$ nodes, we have
\[\sum_{i: \text{ join node}}n_jn_k = O((tn)^2).
\]
Therefore, the total running time of the algorithm is $O(2^t(tn)^2)$.

\begin{theorem}\label{thm:twdp}
    Given a tree decomposition of $G$ of width $t$, {\sc Max Bisection} can be solved in time $O(2^t(tn)^2)$.
\end{theorem}

\subsection{Optimality of our algorithm}\label{ssec:obtimal}
Eiben et al.~\cite{Eiben::2019} proved that if the following {\sc $(\min, +)$-Convolution} does not admit $O(n^{2-\delta})$-time algorithm for some $\delta > 0$,
there is no $O(n^{2-\varepsilon})$-time algorithm for {\sc Min Bisection} on (edge-weighted) trees for any $\varepsilon > 0$.

\begin{definition}[{\sc $(\min, +)$-Convolution}]
    Given two sequences of numbers $(a_i)_{1\le i \le n}$, $(b_i)_{1 \le i \le n}$, the goal of {\sc $(\min, +)$-Convolution} is to compute $c_i = \min_{1 \le j \le i} (a_{j} + b_{i-j+1})$ for all $1 \le i \le n$.
\end{definition}

In fact, they proved that this conditional lower bound holds even on trees with pathwidth two.
In the following, we slightly strengthen their hardness result by modifying their reduction.
\begin{theorem}\label{thm:path}
    Unless {\sc $(\min, +)$-Convolution} is solved in time $O(n^{2 - \delta})$ for some $\delta > 0$,
    there is no $O(n^{2 - \varepsilon})$-time algorithm for {\sc Min Bisection} on (edge-weighted) paths for any $\varepsilon > 0$.
\end{theorem}

\begin{proof}
    We perform a reduction from a variant of {\sc 3SUM} to {\sc Min Bisection} on paths.
    In this variant, we are given three sequences of numbers $(a_i)_{1 \le i \le n}$, $(b_i)_{1 \le i \le n}$, and $(c_i)_{1 \le i \le 2n}$.
    The goal is to compute two indices $i, j$ with $1 \le i, j \le n$ such that $a_i + b_j + c_{i+j} \le 0$ or report that no such pair of indices exists.
    This variant is equivalent to {\sc $(\min, +)$-Convolution} in the sense that one of these problems can be solved in truly subquadratic time, then so is the other one~\cite{Backurs:Better:2017}.
    
    Let $W = 4nM + 1$, where $M$ is the maximum absolute value among $(a_i)_{1 \le i \le n}$, $(b_i)_{1 \le i \le n}$, and $(c_i)_{1 \le i \le 2n}$.
    Let $P_a = (x_1, x_2, \ldots, x_{n+1})$ be a path of length $n$.
    The weight of edge $\{x_{i}, x_{i+1}\}$ is defined as $a_i + W$ for each $1 \le i \le n$.
    Paths $P_b = (y_1, \ldots, y_{n+1})$ and $P_c = (z_1, \ldots, z_{2n+1})$ are defined accordingly.
    The weight of edges $\{y_i, y_{i+1}\}$ and $\{z_i, z_{i+1}\}$ are defined as $b_i + W$ and $c_i + W$, respectively.
    By the definition of $W$, the weight of each edge is positive.
    Now, we construct the entire path $P$ by combining these three paths.
    Let $Q_1, Q_2, Q_3, Q_4$ be four paths of length $100n-1$, $55n-1$, $10n-1$, $55n-1$, respectively.
    Recall that the length of a path is defined by the number of edges.
    For each $1 \le i \le 4$, we let $s_i$ and $t_i$ be the end vertices of $Q_i$.
    Each edge of $Q_i$ has weight $3W + 1$ for all $1 \le i \le 4$.
    The entire path is obtained by (1) connecting $t_1$ to $x_1$ with an edge of weight $3W+1$, (2) identifying $s_2$ and $x_{n+1}$, (3) identifying $t_2$ and $y_{n+1}$, (4) connecting $s_3$ to $y_1$ with an edge of weight $3W + 1$, (4) identifying $t_3$ and $z_{2n+1}$, and (5) connecting $s_4$ to $z_1$ with an edge of weight $3W + 1$, which is illustrated in Figure~\ref{fig:path}.
    \begin{figure}
        \centering
        \includegraphics[width=\textwidth]{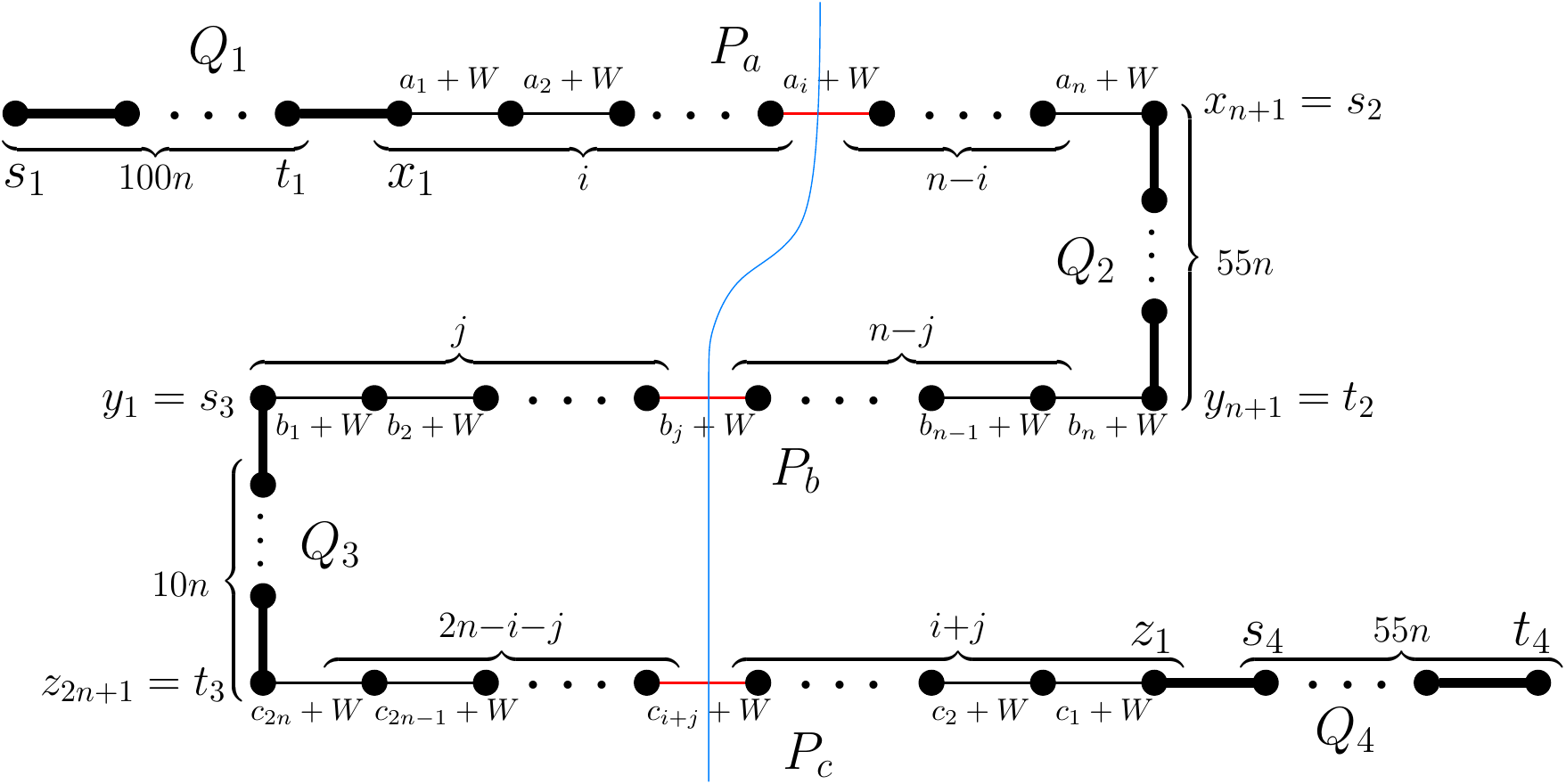}
        \caption{The figure illustrates the construction of the path $P$. Thick lines are edges in $Q_i$.}
        \label{fig:path}
    \end{figure}
    We prove that the given instance of {\sc 3SUM} is feasible if and only if there is a bisection of $P$ of size at most $3W$.
    
    Suppose that there are indices $i, j$ with $1 \le i, j \le n$ such that $a_i + b_j + c_{i+j} \le 0$.
    Then, we construct a cut $(A, B)$ of $P$ in such a way that three edges $\{x_{i}, x_{i+1}\}$, $\{y_{j}, y_{j+1}\}$, and $\{z_{i+j}, z_{i+j+1}\}$ are cut edges.
    More precisely, we let $A = V(Q_1) \cup V(Q_3) \cup \{x_k : 1 \le k \le i\} \cup \{y_k : 1 \le k \le j\} \cup \{z_k : i+j < k \le 2n\}$ and $B = V(P) \setminus A$.
    Since $|A| = |V(Q_1)| + |V(Q_3)| + i + j + 2n-i-j = 112n$ and $|B| = |V(Q_2)| + |V(Q_4)| + n-i+n-j+i+j = 112n$, $(A, B)$ is a bisection of $P$ of size at most $3W + a_i + b_i + c_{i+j} \le 3W$.
    
    Conversely, suppose there is a bisection $(A, B)$ of $P$ of size at most $3W$.
    We first observe that, for each $1 \le i \le 4$, either $V(Q_i) \subseteq A$ or $V(Q_i) \subseteq B$ as otherwise the size of $(A, B)$ exceeds $3W$.
    Without loss of generality, we assume that $V(Q_1) \subseteq A$.
    As $|A| = 112n$, we have $V(Q_2) \cup V(Q_4) \subseteq B$ and then $V(Q_3) \subseteq A$.
    This implies that at least one edge of each of $P_A$, $P_B$, and $P_C$ contributes to bisection $(A, B)$.
    Moreover, if at least four edges of these three paths contribute to the bisection $(A, B)$, the size of $(A, B)$ is more than $3W$.
    This follows from the fact that the sum of the absolute values in $(a_i)_{1\le i\le n}$, $(b_i)_{1 \le i \le n}$, and $c_{1 \le i \le 2n}$ is at most $4nM$, which is smaller than $W$.
    Thus, the bipartition $(A, B)$ separates each path into two parts: 
    \begin{align*}
        (V(P_a) \cap A, V(P_a) \cap B) &= (\{x_1, \ldots, x_i\}, \{x_{i+1}, \ldots, x_{n+1}\}),\\
        (V(P_b) \cap A, V(P_b) \cap B) &= (\{y_1, \ldots, y_j\}, \{y_{j+1}, \ldots, y_{n+1}\}),\\
        (V(P_c) \cap A, V(P_c) \cap B) &= (\{z_{k+1}, \ldots, z_{2n+1}\}, \{z_{1}, \ldots, z_{k}\}).
    \end{align*}
    As $|A| = |B|$, $k$ must be equal to $i + j$ and hence we have $a_i + b_i + c_{i+j} \le 0$.
\end{proof}

This conditional lower bound matches the dependency on $n$ in the running time of our algorithm.

In terms of the dependency on treewidth, we can prove that under the Strong Exponential Time Hypothesis (SETH)~\cite{Impagliazzo::2001}, there is no $(2-\varepsilon)^tn^{O(1)}$-time algorithm for {\sc Max Bisection}, and hence the exponential dependency on $t$ is asymptotically optimal. 
To see this, we use the following result.

\begin{theorem}[\cite{Lokshtanov::2018}]
    Unless SETH fails, there is no algorithm for {\sc Max Cut} on unweighted graphs that runs in time $2^{t - \varepsilon}n^{O(1)}$ for any $\varepsilon > 0$ even if a width-$t$ tree decomposition of the input graph is given as input for some $t$. 
\end{theorem}

The known reduction (implicitly appeared in \cite{Bui::1987}) from {\sc MaxCut} to {\sc Max Bisection} works well for our purpose.
Specifically, let $G$ be an unweighted graph and let $n$ be the number of vertices of $G$.
We add $n$ isolated vertices to $G$ and the obtained graph is denoted by $G'$. 
It is easy to see that $G$ has a cut of size at least $k$ if and only if $G'$ has a bisection of size at least $k$.
Moreover, $\tw(G') = \tw(G)$. Therefore, the conditional lower bound is inherited from {\sc MaxCut}.

\begin{theorem}
    Unless SETH fails, there is no algorithm for {\sc Max Bisection} that runs in time $2^{t - \varepsilon}n^{O(1)}$ for any $\varepsilon > 0$ even if a width-$t$ tree decomposition of the input graph is given as input for some $t$. 
\end{theorem}

\section{A linear-time algorithm on graphs with bounded vertex cover number}\label{sec:vc}
In the previous section, we show that {\sc Min Bisection} does not admit an $O(n^{2 - \varepsilon})$-time algorithm even for path graphs unless {\sc $(\min, +)$-Convolution} can be solved in time $O(n^{2 - \delta})$ for some $\delta > 0$.
To improve the quadratic dependency on $n$, we consider vertex cover number as a graph parameter. 
Recall that {\sc Min Bisection} and {\sc Max Bisection} are equivalent to each other when the edge weight is allowed to be arbitrary.
Thus, in what follows, we consider {\sc Max Bisection}.

We first compute a vertex cover $C$ of $G$ with size at most $k$ in time $O(2^kkn)$ by a standard branching algorithm.
For each $X \subseteq C$, we try to find a maximum size bisection $(A, B)$ of $G$ such that $X \subseteq A$ and $B \cap X = \emptyset$.
Fix $X \subseteq C$.
For $v \in V \setminus C$, we compute a value $p_v = w(C \setminus X, \{v\}) - w(X, \{v\})$.
This can be done in $O(k)$ time for each $v \in V \setminus C$.
Since $V \setminus C$ is an independent set of $G$, we can compute the size of a cut $(A, B)$ as:
\begin{eqnarray*}
    w(A, B) &=& w(X, C \setminus X) + \sum_{v \in A \setminus C} w(C \setminus X, \{v\}) + \sum_{v \in B \setminus C} w(X, \{v\})\\
    &=& w(X, C \setminus X) + \sum_{v \in A \setminus C} p_v + \sum_{v \in V \setminus C} w(X, \{v\}).
\end{eqnarray*}
Since the first and third terms depend only on the choice of $X$, we can compute the maximum size of a bisection by summing up largest $n/2 - |X|$ values of $p_v$ for $v \in V \setminus C$.
To this end, we first compute the $(n/2 - |X|)$-th largest value $t$ among all values $p_v$ for $v \in V \setminus C$.
This can be done in linear time by using a linear-time algorithm for order statistics~\cite{Blum::1973}.
From this threshold $t$, we can easily compute the sum of largest $n/2 - |X|$ values of $p_v$ in linear time.
Therefore, by guessing all $X \subseteq C$, we can solve {\sc Min Bisection} and {\sc Max Bisection} in time $O(2^{k}kn)$.

\begin{theorem}
    {\sc Min Bisection} and {\sc Max Bisection} can be solved in time $O(2^\tau(G)\tau(G)n)$, where $\tau(G)$ is the minimum size of a vertex cover of an input graph $G$. 
\end{theorem}

\section{Hardness on graph classes}\label{sec:hardness}
In this section, we discuss some complexity results for {\sc Min Bisection} and {\sc Max Bisection} on unweighted graphs.
In Section~\ref{ssec:obtimal}, we have seen that there is a quite simple reduction from {\sc MaxCut} to {\sc Max Bisection}.
We formally describe some immediate consequences of this reduction as follows.
Let $\mathcal C$ be a graph class such that 
\begin{itemize}
    \item {\sc MaxCut} is NP-hard even if the input graph is restricted to be in $\mathcal C$ and
    \item for every $G \in \mathcal C$, a graph $G'$ obtained from $G$ by adding arbitrary number of isolated vertices is also contained in $\mathcal C$. 
\end{itemize}
The reduction shows that {\sc Max Bisection} is NP-hard for every graph class $\mathcal C$ that satisfies the above conditions.
\begin{theorem}\label{thm:nph:max}
    {\sc Max Bisection} is NP-hard even for split graphs, comparability graphs, AT-free graphs, and claw-free graphs.
\end{theorem}

It is known that {\sc MaxCut} is NP-hard even for split graphs~\cite{Bodlaender::2000} and comparability graphs~\cite{Pocai::2016}, and co-bipartite graphs~\cite{Bodlaender::2000} which is a subclass of AT-free graphs and claw-free graphs.
If $\mathcal C$ is the class of co-bipartite graphs, the second condition does not hold in general. 
However, we can prove the hardness of {\sc Max Bisection} on co-bipartite graphs, which will be discussed in the last part of this section.

Suppose the input graph $G$ has $2n$ vertices.
Let $\overline{G}$ is the complement of $G$.
It is easy to see that $G$ has a bisection of size at least $k$ if and only if $\overline{G}$ has a bisection of size at most $n^2 - k$.
This immediately gives the following theorem from Theorem~\ref{thm:nph:max}.
\begin{theorem}\label{thm:nph:min}
    {\sc Min Bisection} is NP-hard even for split graphs and co-comparability graphs.
\end{theorem}

For bipartite graphs, {\sc Max Cut} is solvable in polynomial time.
However, we show that  {\sc Min Bisection} and {\sc Max Bisection} are NP-hard even on bipartite graphs.
\begin{theorem}\label{thm:minbisection:bipartite}
    {\sc Min Bisection} is NP-hard even on bipartite graphs.
\end{theorem}
\begin{proof}
We prove the statement by performing a polynomial-time reduction from {\sc Min Bisection} on $4$-regular graphs, which is known to be NP-hard~\cite{Bui::1987}.

Let $G = (V, E)$ be a $4$-regular graph.
We can assume that $G$ has $2n$ vertices since the reduction given by \cite{Bui::1987} works on graphs having even number of vertices.
For each edge $e = \{u, w\} \in E$, we split $e$ by introducing a new vertex $v_e$ and replacing $e$ with two edges $\{u, v_e\}$ and $\{v_e, w\}$.
Then, for each $v \in V$, we add $n^3$ pendant vertices and make adjacent them to $v$.
We denote by $V_E$ the set of vertices newly added for edges, by $V^p$ the set of pendant vertices, and by $G' = (V \cup V_E \cup V^p, E')$ the graph obtained from $G$ as above (see Figure~\ref{fig:reduction}).
As $G$ is $4$-regular, we have $|V_E| = |E| = 4n$ and $|V \cup V_E \cup V^p| = 2n + 4n + 2n^4 = 2n^4+6n$.
Moreover, $G'$ is bipartite.
In the following, we show that $G$ has a bisection of size at most $k$ if and only if so does $G'$.

\begin{figure}[tbp]
    \centering
    \includegraphics[width=8cm]{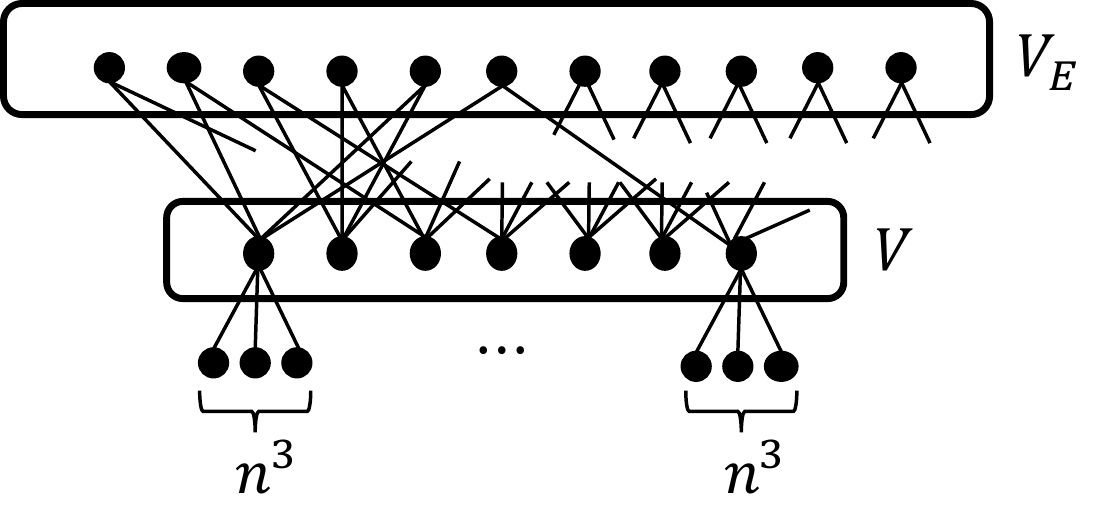}
    \caption{The constructed graph in the proof of Theorem \ref{thm:minbisection:bipartite}.}
    \label{fig:reduction}
\end{figure}

Suppose $G$ has a bisection $(V_1, V_2)$ of size at most $k$.
Since $|V| = 2n$, it holds that $|V_1| = |V_2|$.
For $i = 1, 2$, we set $V'_i = V_i \cup \{v_e : e \subseteq V_i\} \cup V^p_i$, where $V^p_i$ is the set of pendant vertices such that its unique neighbor is contained in $V_i$.
Note that there are no edges between $V'_1$ and $V'_2$ in $G'$ and $|V'_1| = |V'_2|$ so far.
Observe that for every $e \in E(V_1, V_2)$, exactly one of the incidental edges $\{u, v_e\}$ and $\{v_e, w\}$ of the corresponding vertex $v_e$ contributes to its size no matter whether $v_e$ is included in either $V'_1$ or $V'_2$.
Therefore, we can appropriately distribute the remaining vertices $\{v_e : e \in E(V_1, V_2)\}$ to obtain a bisection of size at most $k$.
 
Suppose that $G'$ has a bisection $(V'_1, V'_2)$ of size at most $k$.
Let $V_1 = V'_1 \cap V$ and $V_2 = V'_2$.
We claim that $|V_1| = |V_2|$.
Suppose for contradiction that $|V_1| > |V_2|$.
As $|V_1 \cup V_2| = 2n$, we have $|V_1| \ge n + 1$.
Since $(V'_1, V'_2)$ is a bisection of $G'$,
it holds that $|V'_1| = n^4 + 3n$.
Thus, there are at least $n^3 - 2n + 1$ pendant vertices in $V'_2$ whose neighbors are contained in $V'_1$.
For $n \ge 5$, it holds that $n^3-2n+1 > 4n^2 > k$, contradicting to the assumption that $(V'_1, V'_2)$ is a bisection of size at most $k$.
Moreover, for every edge $e = \{u, w\} \in E(V_1, V_2)$ in $G$, at least one of $\{u, v_e\}$ or $\{v_e, w\}$ contributes to the cut $(V'_1, V'_2)$ in $G'$.
Therefore, we conclude that the size of the cut $(V_1, V_2)$ in $G$ is at most $k$.
\end{proof}

Interestingly, the same construction works well for proving the hardness of {\sc Max Bisection} on bipartite graphs.

\begin{theorem}\label{thm:maxbisection:bipartite}
    {\sc Max Bisection} is NP-hard even on bipartite graphs.
\end{theorem}
\begin{proof}
    The proof of this theorem is similar to that of Theorem~\ref{thm:minbisection:bipartite}.
    Let $G = (V, E)$ be a 4-regular graph and let $G' = (V \cup V_E \cup V^p, E')$ be the bipartite graph described in the proof of Theorem~\ref{thm:minbisection:bipartite}. 
    In the following, we prove that $G$ has a bisection of size at most $k$ if and only if $G'$ has a bisection of size {\em at least} $2n^4 + 8n - k$.
    
    Suppose first that $G$ has a bisection $(V_1, V_2)$ of size $k$.
    Then, we set $V'_i$ for $i \in \{1, 2\}$ as:
    \begin{itemize}
        \item $V_i \subseteq V'_i$;
        \item If $v \in V_i$, all pendant vertices $w$ with $N(w) = \{v\}$ are contained in $V'_{3 - i}$;
        \item For each $e \in E$ with $e \subseteq V_i$, $v_e$ is contained in $V'_{3 - i}$.
    \end{itemize}
    For each remaining $e \in E(V_1, V_2)$, we add $v_e$ to arbitrary side $V'_i$ so that $(V'_1, V'_2)$ becomes a bisection of $G'$.
    This can be done since $G$ is $4$-regular, which means $G[V_i]$ contains exactly $2n-k$ edges for each $i = \{1, 2\}$.
    Let us note that for $v_e \in V_E$ with $e \in E(V_1, V_2)$, exactly one of the incident edges of $v_e$ contributes to the size of the bisection no matter which $V'_i$ includes $v_e$.
    This implies that the size of bisection $(V'_1, V'_2)$ is $2n^4 + 8n - k$.

    For the converse, suppose that $G'$ has a bisection $(V'_1, V'_2)$ of size at least $2n^4 + 8n - k$.
    For each $i = 1, 2$, we let $V_i = V'_i \cap V$. Then, we claim that $(V_1, V_2)$ is a bisection of $G$.
    To see this, we assume for contradiction that $|V_1| > |V_2|$.
    Clearly, $V_1$ contains at least $n + 1$ vertices.
    As $|V'_1| = |V'_2|$ and $G'$ has $2n + 2n\cdot n^3 + 4n = 2n^4 + 6n$ vertices, we have $|V'_2| = n^4 + 3n$. 
    Since $V_1$ has at least $n + 1$ vertices, at least $(n + 1)n^3 - |V'_2| = n^3 - 3n$ pendant vertices adjacent to some vertex in $V_1$ are included in $V_1$.
    Therefore, at most $|E'| - (n^3 - 3n) = 2|E| + 2n \cdot n^3 - (n^3 - 3n) = 2n^4 - n^3 + 12n$ edges can belong to $E'(V'_1, V'_2)$.
    For $n \ge 3$, we have $2n^4 - n^3 + 12n < 2n^4 + 8n -4n < 2n^4 + 8n -k$, which contradicts to the fact that the size of $(V'_1, V'_2)$ is at least $2n^4 + 8n - k$. Note that $k\le 4n$.
    
    Now, we show that the bisection $(V_1, V_2)$ of $G$ has size at most $k$.
    Since there are $2n^4$ pendant edges in $G'$, at least $8n-k$ edges of $G'[V \cup V_E]$ belong to $E'(V'_1, V'_2)$.
    Note that as $V$ and $V_E$ are respectively independent sets in $G'$, such edges are in $E'(V, V_E)$.
    Moreover, there are $8n$ edges in $E'(V, V_E)$.
    If there are at least $k+1$ vertices $v_e$ in $V_E$ having neighbors both in $V_1$ and in $V_2$,
    the size of $E'(V'_1 \cap (V \cup V_E), V'_2 \cap (V \cup V_E))$ is at most $8n-k-1$ since exactly one of the incidental edges of $v_e$ does not contribute to the cut.
    Thus, the number of such vertices is at most $k$. 
    Since each $v_e\in V_E$ having neighbors both in $V_1$ and in $V_2$ corresponds to a cut edge of $(V_1, V_2)$ in $G$, the size of the bisection $(V_1, V_2)$ of $G$ is at most $k$.
\end{proof}

Since both {\sc Min Bisection} and {\sc Max Bisection} are NP-hard on bipartite graphs, by the same argument with Theorem~\ref{thm:nph:min}, we have the following corollary.
\begin{corollary}
    {\sc Min Bisection} and {\sc Max Bisection} are NP-hard even for co-bipartite graphs.
\end{corollary}

\section{Line graphs}
Guruswami \cite{Guruswami::1999} showed that {\sc MaxCut} can be solved in linear time for unweighted line graphs.
The idea of the algorithm is to find a cut satisfying a certain condition using an Eulerian tour of the underlying graph of the input line graph.
In this section, we show that his approach works well for {\sc Max Bisection}.

Let $G = (V, E)$ be a graph. The {\em line graph} of $G$, denoted by $L(G) = (V_L, E_L)$, is an undirected graph with $V_L = E$ such that two vertices $e, f \in V_L$ are adjacent if and only if $e$ and $f$ share a common end vertex in $G$.
We call $G$ an underlying graph of $L(G)$.
Note that from a line graph, its underlying graph is not uniquely determined.
However, it is sufficient to take an arbitrary one of them to discuss our result.
Guruswami gave the following sufficient condition for {\sc MaxCut} and showed that every line graph has a cut satisfying this condition.

\begin{lemma}[\cite{Guruswami::1999}]\label{lem:line:sufficient}
    Let $G = (V, E)$ be a (not necessarily line) graph and let $C_1, C_2, \ldots C_k$ be edge disjoint cliques with $\bigcup_{1\le i \le k}C_i = E$.
    If there is a cut $(A, B)$ of $G$ such that $-1 \le |A \cap C_i| - |B \cap C_i| \le 1$ for every $1 \le i \le k$, then $(A, B)$ is a maximum cut of $G$.
\end{lemma}

Since the maximum size of a bisection is at most the maximum size of a cut, we immediately conclude that every bisection satisfying the condition in Lemma~\ref{lem:line:sufficient} is a maximum bisection.
The construction of a bipartition $(A, B)$ of $V$ in \cite{Guruswami::1999} is as follows.

Let $L(G) = (V_L, E_L)$ be a line graph whose underlying graph is $G$.
We make $G$ an even-degree graph by putting a vertex $r$ and make adjacent $r$ to each vertex of odd degree.
Let $G'$ be the even-degree graph obtained as above.
Suppose first that $G'$ is connected.
Fix an Eulerian tour starting from $r$ and alternately assign labels $a$ and $b$ to each edge along with the Eulerian tour.
Let $A$ and $B$ be the set of edges having label $a$ and $b$, respectively.
Observe that the bipartition $(A \cap V_L, B \cap V_L)$ of $V_L$ satisfies the sufficient condition in Lemma~\ref{lem:line:sufficient}.
To see this, consider a vertex $v$ of $G$.
Since the set of edges $C_v$ adjacent to $v$ forms a clique in $L(G)$.
Moreover, it is known that, in line graphs, the edges of cliques $\{C_v : v \in G\}$ partitions the whole edge set $E_L$.
Every vertex $v$ of $G'$ except for $r$ has an equal number of incidental edges with label $a$ and those with label $b$ in $G'$, which implies that every clique $C_v$ satisfies $-1 \le |C_v \cap (A \cap V_L)| - |C_v \cap (B \cap V_L)| \le 1$.

Now, we show that the bipartition $(A \cap V_L, B \cap V_L)$ of $V_L$ is also a bisection of $L(G)$.
Consider the labels of the edges incident to $r$ in $G'$.
Observe that every two consecutive edges in the Eulerian tour except for the first and last edge have different labels.
Moreover, the first edge has label $a$.
If the last edge has label $a$, we have $|A| = |B| + 1$ and hence $|A \cap V_L| + 1 = |B \cap V_L|$.
Otherwise, the last edge has label $b$, we have $|A| = |B|$ and hence $|A \cap V_L| = |B \cap V_L|$.
Therefore, $(A \cap V_L, B \cap V_L)$ is a bisection of $L(G)$.

If $G'$ has two or more connected components, we apply the same argument to each connected component and appropriately construct a bipartition of $V_L$.
It is not hard to see that this bipartition also satisfies the condition in Lemma~\ref{lem:line:sufficient}.

Since, given a line graph, we can compute its underlying graph~\cite{Lehot::1974,Roussopoulos::1973} and an Eulerian tour in linear time,
{\sc Max Bisection} on line graphs can be solved in linear time. 

\begin{theorem}
    {\sc Max Bisection} can be solved in linear time on line graphs.
\end{theorem}

\section{Conclusion}
In this paper, we show that there is an $O(2^{t}(tn)^2)$-time algorithm for solving {\sc Min Bisection} and {\sc Max Bisection}, provided that a width-$t$ tree decomposition is given as input.
This running time matches the conditional lower bound given by Eiben et al.~\cite{Eiben::2019} based on {\sc $(\min, +)$-Convolution}.
We slightly strengthen their conditional lower bound by exhibiting a reduction from {\sc $(\min, +)$-Convolution} to {\sc Min Bisection} on paths.
The exponential dependency of treewidth in our running time has been shown to be asymptotically optimal under the Strong Exponential Time Hypothesis.

For unweighted graphs, Eiben et al. showed that the quadratic dependency in the running time can be slightly improved: They gave an $O(8^tt^{O(1)}n^{1.864}\log n)$-time algorithm for {\sc Min Bisection} using an extension of the fast $(\min, +)$-convolution technique due to Chan et al.~\cite{Chan::2015}.
It would be interesting to know whether a similar improvement can be achieved in our case.

We also show that {\sc Min Bisection} and {\sc Max Bisection} are NP-hard even for several restricted graph classes.
In particular, both problems are NP-hard even on unweighted bipartite graphs, which is in contrast with the tractability of {\sc MinCut} and {\sc MaxCut} on this graph class.
There are several open problems regarding special graph classes.
One of the most notable open questions would be to reveal the complexity of {\sc Min Bisection} on planar graphs.

\section*{Acknowledgments}
The authors thank anonymous reviewers for giving us valuable comments to the preliminary version of this paper~\cite{Hanaka:Optimal:2020}.
This work is partially supported by JSPS KAKENHI Grant Numbers JP19K21537, JP17H01788 and JST CREST JPMJCR1401.

\bibliographystyle{plain}
\bibliography{main}

\begin{thebibliography}{10}

\bibitem{Backurs:Better:2017}
Arturs Backurs, Piotr Indyk, and Ludwig Schmidt.
\newblock Better approximations for tree sparsity in nearly-linear time.
\newblock In {\em Proceedings of the Twenty-Eighth Annual ACM-SIAM Symposium on
  Discrete Algorithms}, SODA ’17, page 2215–2229, USA, 2017. Society for
  Industrial and Applied Mathematics.

\bibitem{Blum::1973}
Manuel Blum, Robert~W. Floyd, Vaughan Pratt, Ronald~L. Rivest, and Robert~E.
  Tarjan.
\newblock Time bounds for selection.
\newblock {\em Journal of Computer and System Sciences}, 7(4):448 -- 461, 1973.

\bibitem{Bodlaender::1998}
Hans~L. Bodlaender and Torben Hagerup.
\newblock Parallel algorithms with optimal speedup for bounded treewidth.
\newblock {\em SIAM Journal on Computing}, 27(6):1725--1746, 1998.

\bibitem{Bodlaender::2000}
Hans~L. Bodlaender and Klaus Jansen.
\newblock On the complexity of the maximum cut problem.
\newblock {\em Nordic J. of Computing}, 7(1):14–31, 2000.

\bibitem{Bui::1987}
Thang~Nguyen Bui, F.~Thomson Leighton, Soma Chaudhuri, and Michael Sipser.
\newblock Graph bisection algorithms with good average case behavior.
\newblock {\em Combinatorica}, 7(2):171–191, 1987.

\bibitem{Chan::2015}
Timothy~M. Chan and Moshe Lewenstein.
\newblock Clustered integer 3sum via additive combinatorics.
\newblock In {\em Proceedings of the Forty-Seventh Annual ACM Symposium on
  Theory of Computing}, STOC ’15, page 31–40, New York, NY, USA, 2015.
  Association for Computing Machinery.

\bibitem{Diaz::2007}
Josep D\'{\i}az and Marcin Kamiski.
\newblock {MAX-CUT and MAX-BISECTION Are NP-Hard on Unit Disk Graphs}.
\newblock {\em Theor. Comput. Sci.}, 377(1–3):271–276, 2007.

\bibitem{Diaz::2017}
Josep D\'iaz and George~B. Mertzios.
\newblock Minimum bisection is {NP}-hard on unit disk graphs.
\newblock {\em Information and Computation}, 256:83 -- 92, 2017.

\bibitem{Eiben::2019}
Eduard Eiben, Daniel Lokshtanov, and Amer~E. Mouawad.
\newblock {Bisection of Bounded Treewidth Graphs by Convolutions}.
\newblock In Michael~A. Bender, Ola Svensson, and Grzegorz Herman, editors,
  {\em 27th Annual European Symposium on Algorithms (ESA 2019)}, volume 144 of
  {\em Leibniz International Proceedings in Informatics (LIPIcs)}, pages
  42:1--42:11, Dagstuhl, Germany, 2019. Schloss Dagstuhl--Leibniz-Zentrum fuer
  Informatik.

\bibitem{Garey::1976}
M.R. Garey, D.S. Johnson, and L.~Stockmeyer.
\newblock Some simplified {NP}-complete graph problems.
\newblock {\em Theoretical Computer Science}, 1(3):237 -- 267, 1976.

\bibitem{Guruswami::1999}
Venkatesan Guruswami.
\newblock Maximum cut on line and total graphs.
\newblock {\em Discrete Appl. Math.}, 92(2–3):217–221, 1999.

\bibitem{Hadlock::1975}
F.~Hadlock.
\newblock Finding a maximum cut of a planar graph in polynomial time.
\newblock {\em SIAM Journal on Computing}, 4(3):221--225, 1975.

\bibitem{Hanaka:Optimal:2020}
Tesshu Hanaka, Yasuaki Kobayashi, and Taiga Sone.
\newblock An optimal algorithm for {\sc bisection} for bounded-treewidth
  graphs.
\newblock In {\em Proceedings of 14th International Frontiers of Algorithmics
  Workshop}, 2020.
\newblock Accepted.

\bibitem{Impagliazzo::2001}
Russell Impagliazzo and Ramamohan Paturi.
\newblock On the complexity of $k$-sat.
\newblock {\em J. Comput. Syst. Sci.}, 62(2):367–375, 2001.

\bibitem{Jansen::2005}
Klaus Jansen, Marek Karpinski, Andrzej Lingas, and Eike Seidel.
\newblock {Polynomial Time Approximation Schemes for MAX-BISECTION on Planar
  and Geometric Graphs}.
\newblock {\em SIAM Journal on Computing}, 35(1):110--10, 2005.

\bibitem{Karp::1972}
Richard~M Karp.
\newblock Reducibility among combinatorial problems.
\newblock In {\em Complexity of computer computations}, pages 85--103.
  Springer, 1972.

\bibitem{Kloks::1994}
Ton Kloks.
\newblock {\em Treewidth, Computations and Approximations}, volume 842 of {\em
  Lecture Notes in Computer Science}.
\newblock Springer, 1994.

\bibitem{Lehot::1974}
Philippe G.~H. Lehot.
\newblock An optimal algorithm to detect a line graph and output its root
  graph.
\newblock {\em J. ACM}, 21(4):569–575, 1974.

\bibitem{Lokshtanov::2018}
Daniel Lokshtanov, D\'{a}niel Marx, and Saket Saurabh.
\newblock Known algorithms on graphs of bounded treewidth are probably optimal.
\newblock {\em ACM Trans. Algorithms}, 14(2), 2018.

\bibitem{Orlova::1972}
G.~I. Orlova and Y.~G. Dorfman.
\newblock Finding the maximal cut in a graph.
\newblock {\em Engineering Cybernetics}, 10:502--506, 1972.

\bibitem{Pocai::2016}
Rafael~Veiga Pocai.
\newblock {The Complexity of SIMPLE MAX-CUT on Comparability Graphs}.
\newblock {\em Electronic Notes in Discrete Mathematics}, 55:161 -- 164, 2016.

\bibitem{Roussopoulos::1973}
Nicholas~D. Roussopoulos.
\newblock A max $\{m,n\}$ algorithm for determining the graph ${H}$ from its
  line graph ${G}$.
\newblock {\em Information Processing Letters}, 2(4):108 -- 112, 1973.

\bibitem{Shih::1990}
W.-K. Shih, S.~Wu, and Y.~S. Kuo.
\newblock Unifying maximum cut and minimum cut of a planar graph.
\newblock {\em IEEE Trans. Comput.}, 39(5):694–697, 1990.

\end{thebibliography}

\end{document}